\documentclass[letterpaper, 10 pt, conference]{ieeeconf}

\IEEEoverridecommandlockouts                              
\usepackage[utf8]{inputenc}

\let\labelindent\relax
\usepackage{amssymb,amsmath,amsthm,mathrsfs}
\usepackage[dvipsnames]{xcolor}
\usepackage{soul}
\usepackage{times,cite}
\usepackage{graphicx}
\usepackage{dblfloatfix}
\usepackage{caption}
\usepackage{subcaption}
\usepackage[shortlabels]{enumitem}
\newtheorem*{proposition*}{Proposition}
\newtheorem*{theorem*}{Theorem}
\newtheorem{assumption}{Assumption}
\newtheorem{lemma}{Lemma}
\newtheorem{theorem}{Theorem}
\newtheorem{definition}{Definition}

\newtheorem{remark}{Remark}
\newtheorem{proposition}{Proposition}
\newcommand{\sanote}[1]{\textcolor{red}{[SA: #1]}}

\usepackage{tikz}
\usetikzlibrary{arrows}
\usetikzlibrary{calc}

\DeclareMathOperator*{\argmin}{argmin}

\title{\LARGE \bf Linear Quadratic Mean-Field Games with Communication Constraints} 

\author{Shubham~Aggarwal, Muhammad~Aneeq~uz~Zaman, and Tamer~Ba{\c s}ar
\thanks{The authors are affiliated with the Coordinated Science Lab, University of Illinois at Urbana-Champaign, Urbana, IL, USA 61801. Emails:
        {\tt\small \{sa57,mazaman2,basar1\}@illinois.edu}.}
\thanks{Research is supported in part by an AFOSR Grant (FA9550-19-1-0353).
}%
}

\tikzset{ remember picture,
   switch/.style = {rectangle,
                    draw,align=center,
                    label={below:#1},
   },
}

\newsavebox\mybox
\savebox\mybox{%
\tikz\draw[line width=0.7pt] (-0.4,0)--(0,0)
								(0,0)--(0.4,0.4);%
}
\begin{document}
\tikzstyle{rect} = [draw,rectangle,fill = white!20,minimum width = 3pt, inner sep  = 5pt]
\tikzstyle{line} = [draw, -latex]
\tikzstyle{dline} = [draw, dotted, -latex]

\maketitle
\thispagestyle{empty}
\pagestyle{empty}
\begin{abstract}
In this paper, we study a large population game with heterogeneous dynamics and cost functions solving a consensus problem. Moreover, the agents have communication constraints which appear as: (1) an Additive-White Gaussian Noise (AWGN) channel, and (2) asynchronous
data transmission via a fixed scheduling policy. Since the complexity of solving the game increases with the number of agents, we use the Mean-Field Game paradigm to solve it. Under standard assumptions on the information structure of the agents, we prove that the control of the agent  in the MFG setting is free of the dual effect. This allows us to obtain an equilibrium control policy for the generic agent, which is a function of only the local observation of the agent. Furthermore, the equilibrium mean-field trajectory is shown to follow linear dynamics, hence making it computable. We show that in the finite population game, the equilibrium control policy prescribed by the MFG analysis constitutes an $\epsilon-$Nash equilibrium, where $\epsilon$ tends to zero as the number of agents goes to infinity. The paper is concluded with simulations demonstrating the performance of the equilibrium control policy.
\end{abstract}
\section{Introduction}
In distributed real-world applications, like networked control systems \cite{ramesh2013design}, ecosystem monitoring \cite{gao2018optimal}, and energy harvesting \cite{nayyar2013optimal}, we rarely have the luxury of pure persistent communication. Hence, in this work, we study multi-agent
systems under a constrained communication structure. The communication constraints may appear in the form of limited sensor energy levels \cite{nayyar2013optimal}, noisy transmission medium \cite{tatikonda2004stochastic,tatikonda2004control}, limits on the communication frequency \cite{imer2010optimal,gao2015optimal} or some combination thereof, as we investigate in this work. Additionally, scalability becomes a huge challenge with increasing number of agents in a multi-agent system.
\par In this paper, we consider a discrete-time multi-agent game problem. Each agent is coupled with other agents through its cost function, which  incentivizes  the  agent  to  form  consensus with other players. In addition to a plant and a controller, the agent's control system (See Figure \ref{Inf_flow}) also consists of a scheduler and an Additive White Gaussian Noise (AWGN) channel. The scheduler controls the flow of information using a fixed scheduling policy. The communication through the AWGN channel is regulated by an encoder/decoder pair which constitutes a predictive encoder \cite{tatikonda1998control} to encode sequential data and a minimum mean-square estimation (MMSE) decoder to produce the best estimate of the plant state.

\medskip

\noindent \textbf{Related Work:} There have been several works in the literature studying estimation and control problems under communication constraints. Reference \cite{bansal1989simultaneous} considers simultaneous design of measurement and control strategies for a class of Linear Quadratic Gaussian (LQG) problems under soft constraints on both. The LQG problem has been further studied for a noisy analog channel \cite{tatikonda1998control} and a noiseless digital channel \cite{tatikonda2004stochastic} in the forward loop. All these works however, consider single agent problems with uninterrupted communication, unlike the setting of this work. In \cite{ramesh2013design}, the authors consider a problem where a network of plants share a noiseless communication medium via a state-based scheduling policy. The system has been shown to be dual effect free under a symmetry condition on the scheduling policy. Similarly the optimality of certainty equivalent control laws has been characterized under an event-triggered communication policy with a noiseless channel in \cite{molin2012optimality,antunes2019consistent}. Our work on the other hand proposes a multi-agent game where each agent has intermittent access to its state measurement through a noisy channel.

\par A key difficulty in multi-agent systems is that of scalability. To alleviate this challenge, a mean-field game (MFG) framework was proposed in \cite{huang2006large, huang2007large} by Huang, Malham{\'e} and Caines and simultaneously in \cite{lasry2007mean}, by Lasry and Lions. The essential idea in the MFG framework is that as the number of agents goes to infinity, agents become indistinguishable and the effect of
individual deviation becomes negligible (that is, the effect of strategic interaction disappears). This leads to an aggregation effect, which can be modelled by an exogenous mean-field (MF) term. Consequently, the game problem reduces to a stochastic optimal control problem for a representative agent along with a consistency condition.
\par Linear Quadratic MFGs (LQ-MFGs), which combine linear agent dynamics with a quadratic cost function, serve as a significant benchmark in the study of MFGs. Recent works on LQ-MFGs \cite{uz2020approximate, uz2020reinforcement} in the discrete-time setting are free of communication constraints or consider partially observed dynamics involving packet drop-outs \cite{moon2014discrete}, thereby making the underlying communication link unreliable. Furthermore, Secure MFGs \cite{uz2020secure} capture the setting where the agents deliberately obfuscate their state information with the goal of subverting an eavesdropping adversary. In these works, however, communication occurs at every time instance, in contrast to our setting here, where the communication is intermittent and the channel adds noise to the incoming signal.

\noindent \textbf{Contribution:} In this paper, we prove that under a fixed scheduling policy, an AWGN channel and a standard information structure, the dual effect of control \cite{bar1974dual} does not show up. The result is presented in Lemma \ref{L1} and is one of the key the observations of the paper. This renders the covariance of estimation error independent of control signals (for both transmission and non-transmission times). Under the mean-field setting, this insight enables us to reduce the game to solving a standard optimal control tracking problem \cite{uz2020reinforcement} along with a consistency condition. We prove the consistency condition of the mean-field equilibrium (MFE) under standard assumptions and characterize the linear dynamics of the equilibrium MF trajectory. Finally, we prove that the policies prescribed by the MFE constitute an $\epsilon$-Nash equilibrium for the finite population game and provide simulations to illustrate the performance of the equilibrium control policy.

The paper is organized as follows. Following this introduction, Section \ref{Sec2} introduces the finite-agent game formulation of the multi-agent system and the underlying information structures of each of the its entities (see Fig. \ref{Inf_flow}). In Section \ref{Sec3}, we formulate the LQ-MFG problem, characterize its MFE and demonstrate the $\epsilon$-Nash property of the MFE. In Section \ref{Sec4}, we provide simulations to analyze the performance of the MFE and conclude the paper in Section \ref{Sec5} sharing some highlights.

\textbf{Notations: }Let $X_k^i$ denote the $i^{th}$ agent's state at time instance $k$ and $X_{k:k'}^i$ the $i^{th}$ agent's state history from instant $k$ to $k'$, i.e., $X_{k:k'}^i = (X_{k}^i, \cdots, X_{k'}^i)$.  Let the set of non-negative integers and real numbers be denoted by $\mathbb{Z}^+$ and $\mathbb{R}^+$, respectively. The transpose of matrix $A$ is denoted by $A'$ and trace of a square matrix $M$ by $Tr\{M\}$. For some vector $z$ and positive semi-definite matrix $S$, let $\|z\|^2_S = z'Sz$. Unless stated otherwise, let $\|\cdot\|$ denote the 2-norm.

\section{Problem Formulation}\label{Sec2}
Consider an $N$-player game on infinite time horizon. Each agent's dynamics evolves according to a linear discrete-time controlled stochastic process as
\begin{align}\label{system}
X_{k+1}^i = A(\phi_i)X_k^i + B(\phi_i)U_k^i + W_k^i, ~~i \in [1,N],
\end{align}
where $X_k^i \in \mathbb{R}^n$ and $U_k^i \in \mathbb{R}^m$ are the state process and the control input, respectively, for the $i^{th}$ agent. $ W_k^i \in \mathbb{R}^n$ is an i.i.d. Gaussian process with zero mean and finite covariance $\Sigma_w$. The initial state $X_0^i$ has mean $\nu_{\phi_i,0}$ and covariance $\Sigma_x$, and is assumed to be statistically independent of $W_k^i$, $ \forall k \in \mathbb{Z}^+$. All covariance matrices are assumed to be positive definite. $A(\phi_i)$ and $B(\phi_i)$ are constant matrices with appropriate dimensions. $\phi_i$ denotes the type of the $i^{th}$ agent drawn from a finite set $\Phi: = \{\phi_1, \cdots, \phi_m\}$ and is chosen according to the empirical distribution 
\begin{align}\label{Emp_D}
F_N(\phi) = \frac{1}{N}\sum_{i=1}^{N}\mathbb{I}_{\{\phi_i \leq \phi\}}, ~~\phi \in \Phi ,
\end{align} where $\mathbb{I}_{\{\cdot\}}$ is the indicator function. It is further assumed that $\lim\limits_{N \rightarrow \infty}{F_N(\phi)} = F(\phi)$ weakly, for some probability distribution $F(\phi)$ over the support of $\Phi$, with corresponding probability mass functions $P_N(\phi)$ and $P(\phi)$, respectively.

\par To complete the problem formulation, we define the information structure on each of the blocks in Fig. \ref{Inf_flow}. Such information structures are standard and appear in applications like industrial and process control \cite{ramesh2013design} and wireless sensor networks \cite{imer2010optimal}. 
\begin{table*}[tbp]
\centering
\begin{tabular}{|c|l|c|l|}
\hline
\textbf{Entity} & \textbf{Information State} & \textbf{Information Space} & \textbf{Input-output Map}\\
\hline
Scheduler & $I^{i,sc}_k \triangleq \left(X_{0:k}^i,Y_{0:k-1}^i\right)$ &  $\mathcal{I}^{i,sc}_k$ & $\xi_k^i: ~I^{i,sc}_k \mapsto \gamma_k$ \\ \hline
Encoder & $I_k^{i,\epsilon}  \triangleq \left(X_{0:k}^i,U_{0:k-1}^i,c_{\mathcal{K}(k-1)}^i,d_{\mathcal{K}(k-1)}^i,Y_{0:k-1}^i \right)$ & $\mathcal{I}_k^{i,\epsilon}$ & $\mathscr{E}_k^i: ~I_k^{i,\epsilon} \mapsto c_{k}^i.$\\ \hline
Decoder & $I_k^{i,d} \triangleq \left(Y_{0:k-1}^i,U_{0:k-1}^i,d_{\mathcal{K}(k)}^i \right)$ & $\mathcal{I}_k^{i,d}$ & $\mathscr{D}_k^i: ~I_k^{i,d} \mapsto Y_k^i$\\ \hline
Controller & $I_k^{i,\pi} \triangleq \left(U_{0:k-1}^i,Y_{0:k}^i \right)$ & $\mathcal{I}_k^{i,\pi}$ & $\pi_k^i: ~I_k^{i,\pi} \mapsto U_k^i$\\
\hline
\end{tabular}
\caption{Dictionary of Information states of entities in Fig. \ref{Inf_flow}}
\label{tab:Infostates}
\hrule
\end{table*}
First, we define a transmission time as an instant when information is sent over the channel. Let the history of transmission times up to the current instant $k$ ($k>0$) be denoted by the set $\mathcal{K}(k):= \{l \mid l \leq k, \gamma_l = 1\}$, where $\gamma_l$ denotes the transmission instant as formalized in the next paragraph. By convention, we take $\mathcal{K}(0)=\{0\}$. 
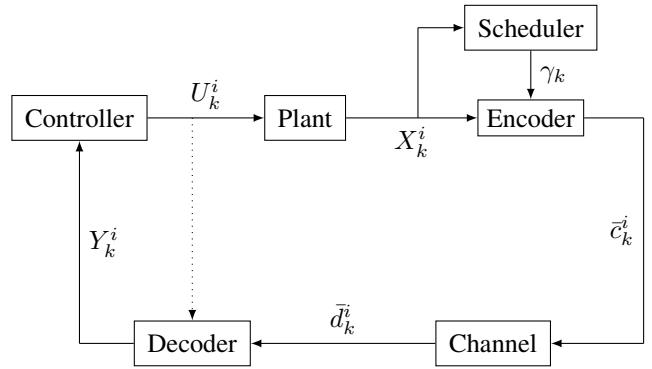
\begin{figure}[t]
    \centering
        \begin{tikzpicture}[node distance = 3cm,auto]
        \node[rect] (C) {Controller};
        \node[rect, right of = C] (P) {Plant};
        \node[switch, right of = P] (E) {Encoder};
        \node[right of=E, xshift=-1.5cm] (Aux1){};
        \node[below of=Aux1] (Aux2){};
        \node[below of=C] (Aux4){};
        \node[rect, above of = E, yshift=-1.8cm] (S) {Scheduler};
        \node[left of=S, xshift= 1.5cm] (Aux3){};
        \node[rect, below of = E, xshift=-0.5cm] (Ch) {Channel};
        \node[rect, left of = Ch, xshift=-1cm] (D) {Decoder};
        \path[line] (C) -- node[midway,above] {$U_k^i$} (P);
        \path[line] (Aux3.center) -- (S);
        \draw ($(P)!.5!(E)$) -- (Aux3.center);
        \draw [line] (P) -- node[midway,below] {$X_k^i$} (E);
        \draw (E) -- (Aux1.center);
        \draw (Aux1.center) -- node[midway,left] {$\bar{c}_k^i$} (Aux2.center);
        \path[line] (S) -- node[midway,right] {$\gamma_k$} (E);
        \path[line] (Aux2.center) -- (Ch);
        \path[line] (Ch) -- node[midway,above] {$\bar{d}_k^i$} (D);
        \draw (D.west) -- (Aux4.center);
        \path [dline] ($(C)!.5!(P)$) -- (D);
        \path[line] (Aux4.center) -- node[midway,right] {$Y_k^i$} (C);

        \end{tikzpicture}
        \caption{Closed-loop information flow for the $i^{th}$ agent	}
        \label{Inf_flow}
\end{figure}
\par The information states, information spaces, and the input-output maps defining the scheduler, encoder, decoder and controller are as defined in Table \ref{tab:Infostates}.
The scheduler has access to the history of plant states and the decoded outputs, based on which, it decides the transmission times of plant state through the channel. The decision whether to transmit or not is taken based on an (innovations based) threshold scheduling policy (of the form $\delta_k'S\delta_k \geq \alpha$, where $\delta_k$ is the error between plant and decoder output at instant $k$, to be defined later, $S > 0$ is a user-defined constant positive definite matrix and $\alpha > 0$ is the threshold parameter). Note that here, by innovations-based process, we mean, given a process $\{z_k\}$, the innovation process $\{\tilde{z}_k\}$ contains new information not carried in the sequences $z_{k-1}, z_{k-2}, \cdots$ \cite[Section 5.3]{anderson2012optimal}. We define $\gamma_k := \xi_k^i\left(I^{i,sc}_k\right)$, where $\gamma_k = 1$ signifies that $k$ is a transmission instant and $\gamma_k = 0$ signifies no transmissions ($\varphi$).

\par Next, the encoder transmits encoded state information ($c_k^i$) at transmission times over the channel. The signal $\bar{c}_k^i$ in Figure \ref{Inf_flow} is given as:
\begin{align*}
\bar{c}_k^i = \begin{cases}
c_k^i, & \text{if} ~~\gamma_k = 1, \\
\varphi, & \text{if} ~~\gamma_k = 0.
\end{cases}
\end{align*}
The encoder is assumed to have full knowledge of the system as in Table \ref{tab:Infostates}, which leads to better control performance compared to structures where only partial information is available \cite{tatikonda2004stochastic} (cf. \cite{tatikonda2004stochastic} for a detailed study on other encoder information structures and how they may be realized). Such a situation emerges when the encoder and the scheduler are collocated with the plant, and hence, can observe both its state as well as the control actions applied to it \cite{tatikonda2004stochastic}. In addition, we assume that the encoder is predictive, i.e., it transmits over the channel, functions of the true state minus the decoder output at the same instant. Such encoders are used in practice while encoding sequential data \cite{tatikonda1998control}.

It is imperative to note here that the above assumptions on the encoder and the scheduler information structures, as we prove in the next section, will entail no dual-effect of control, and hence, lead to a simple controller design. The dual effect of control refers to the dual role of the controller in the evolution of the system dynamics and to probe the scheduler and new measurements to reduce its uncertainty on the system state \cite{molin2012optimality}. Finally, we note, in the table, $c_{\mathcal{K}(k)}^i := \{c_l^i \mid l \leq k, \gamma_l = 1\}$ and $d_{\mathcal{K}(k)}^i := \{d_l^i \mid l \leq k, \gamma_l = 1\}$.

\par The encoded signal is sent over an AWGN channel, which is analog, memoryless and is modeled as:
\begin{align*}
d_{l}^i = c_{l}^i + v_{l}^i,~ l \in \mathcal{K}(k)
\end{align*}
where $v_{l}^i$ is an i.i.d. zero mean Gaussian process with finite positive-definite covariance $\Sigma_v$ and represents the channel noise.
The input and output alphabets of the channel lie in $\mathbb{R}^n$. The signal $\bar{d}_k^i$ in Figure \ref{Inf_flow} is given as:
\begin{align*}
\bar{d}_k^i = \begin{cases}
d_k^i, & \text{if} ~~\gamma_k = 1, \\
\varphi, & \text{if} ~~\gamma_k = 0.
\end{cases}
\end{align*}


\noindent Next, the decoder at the controller end serves two purposes. Firstly, it decodes the noisy channel output to produce a MMSE estimate \cite{tatikonda1998control} of the input signal whenever \textit{new} information is received via the channel. Secondly, between transmission times, it calculates a recursive estimate of the plant state to send information to the controller at all times $k$. Thus, the complete decoder mapping is given by
\begin{align}\label{Recursive}
Y_k^i = \begin{cases}
\mathbb{E}\{X_k^i|I_k^{i,d}\}, & \text{if} ~~\gamma_k = 1, \\
AY_{k-1}^i + BU_{k-1}^i, & \text{if} ~~\gamma_k = 0,
\end{cases}
\end{align} where $AY_{k-1}^i + BU_{k-1}^i$ is the recursive estimate calculated by the decoder between transmission instants, and  $Y_k^i$ is the input to the controller.

\par Finally, the controller calculates control actions by minimizing an infinite-horizon average cost function
\begin{align}\label{LQT}
&J_i^N(\pi^i,\pi^{-i}) :=  \\& \limsup_{T\rightarrow \infty}{\frac{1}{T}\mathbb{E}\left\lbrace \sum_{k=0}^{T-1}\|X_k^i-\frac{1}{N}\sum_{j=1}^{N}{X_k^j}\|^2_{Q(\phi_i)} + \|U_k^i\|^2_{R(\phi_i)}\right\rbrace} \nonumber
\end{align}
where $Q(\phi_i) \geq 0$, $R(\phi_i)>0$ and the parameter $\phi_i \in \Phi$ determines the tuple $(A(\phi), B(\phi), Q(\phi),R(\phi))$ for each agent. Further, $\pi^{-i}:= (\pi^1,\cdots, \pi^{i-1}, \pi^{i+1}, \cdots \pi^N)$, where $\pi^i_k$ is as defined in Table \ref{tab:Infostates}. The control law  for agent $i$ over the sequence of \emph{deterministic} control policies $\pi^i := (\pi_0^i,~\pi_1^i, \cdots, ) \in \bar{\mathcal{M}}_i$ and $\bar{\mathcal{M}}_i$ is the space of admissible decentralized control laws. 
The coupling between agents enters via the consensus term $\frac{1}{N}\sum_{j=1}^{N}{X_k^j}$ in the objective. Further, the cost incorporates a soft constraint on the control actions alongside penalizing state deviations from the consensus term, which each agent aims to track. Finally, the expectation in \eqref{LQT} is taken with respect to noise statistics and the initial state distribution.
\par Now, that the problem description is complete, the aim is to design a decentralized control policy for each agent in \eqref{system} minimizing its local objective \eqref{LQT}, which is done using the MFG framework in the next section.
\section{Mean-Field Games}\label{Sec3}
In this section, we solve the problem for the $N$-player system \eqref{system} with objective \eqref{LQT} by considering the limiting case as $N \rightarrow \infty$. In this setting, the consensus term in \eqref{LQT} can be approximated by a known deterministic sequence (also termed the mean-field trajectory) following  the  Nash  Certainty  Equivalence  Principle \cite{huang2006large}. This reduces the problem to a tracking control problem and a consistency condition as shown later. First, we obtain the solution (to a fully observed tracking problem constructed from a partially observed one) for this infinite agent (mean-field) system. This solution, called the MFE, consists of computing an equilibrium control policy and the equilibrium MF trajectory. Finally, we demonstrate its $\epsilon$--Nash property.
\subsection{Optimal Tracking Control}
Consider a generic agent (from an infinite population system) of type $\phi$ with dynamics 
\begin{align}\label{systemMF}
X_{k+1} = A(\phi)X_k + B(\phi)U_k + W_k, ~k \in \mathbb{Z}^+
\end{align}
where $X_k \in \mathbb{R}^n$ and $U_k \in \mathbb{R}^m$ are the state process and the control input, respectively. The initial state $X_0$ has mean $\nu_{\phi,0}$ and finite positive-definite covariance $\Sigma_x$. Further, $W_k \in \mathbb{R}^n$ is an i.i.d. Gaussian process with zero mean and finite positive-definite covariance $\Sigma$.
Let us denote the generic agent's controller information space at time $k$ as $\mathcal{I}_k^{\mu}$. Then, its information state at any time $k$ is $I_k^{\mu} \triangleq (U_{0:k-1},Y_{0:k}) \in \mathcal{I}_k^{\mu}$. Let us define the map
$
\mu_k: ~ \mathcal{I}_k^{\mu} \rightarrow \mathcal{U},
$ or more specifically, $\mu_k$ maps $I_k^{\mu}$ to $U_k$. 
The control law can then be given as $\mu:=(\mu_0, \mu_1\cdots,) \in \mathcal{M}$, where $\mathcal{M}$ denotes the admissible class of control laws. The objective function for the generic agent can be given as
\begin{align}\label{LQTMF}
&J(\mu,\bar{X}) \\& := \limsup_{T\rightarrow \infty}{\frac{1}{T}\mathbb{E}\left\lbrace \sum_{k=0}^{T-1}\|X_k-\bar{X}_k\|^2_{Q(\phi)} + \|U_k\|^2_{R(\phi)}\right\rbrace} \nonumber
\end{align}
where the expectation is taken over the noise statistics, initial state and the joint laws $\mu$. In addition, $\bar{X} = (\bar{X}_1,\bar{X}_2, \cdots)$ is the MF trajectory, which represents the infinite player approximation to the coupling term in \eqref{LQT}. The introduction of this term leads to indistinguishability between the agents, thereby making the effect of state deviations of individual agents negligible. Consequently, the game problem reduces to a stochastic optimal control problem for the generic agent followed by a consistency condition, whose solution is given by the MFE. Before we formally define the MFE, we state the following assumption on the mean-field system.
\begin{assumption}\label{Ass3}
\begin{enumerate}[(i)]
\item The pair $(A(\phi),B(\phi))$ is controllable and $(A(\phi),{Q(\phi)}^{\frac{1}{2}})$ is observable.
\item The MF trajectory $\bar{X} \in \mathcal{X}$, where $\mathcal{X} := \{\bar{X}_k \in \mathbb{R}^n : \|\bar{X}\|_{\infty}:= \sup_{k\geq 0}{\|\bar{X}_k\|} < \infty\}$ is the space of bounded vector-valued functions.
\end{enumerate}
\end{assumption}

\par Now, to define an MFE, we introduce two operators \cite{uz2020reinforcement}:
\begin{enumerate}
\item $\Psi: \mathcal{X} \rightarrow \mathcal{M}$ given by $\Psi(\bar{X}) = \argmin_{\mu \in \mathcal{M}}J(\mu,\bar{X})$, which outputs the optimal control policy minimizing the cost \eqref{LQTMF}, and

\item $\Lambda : \mathcal{M \rightarrow \mathcal{X}}$ given by $\Lambda(\mu) = \bar{X}$, also called the consistency operator, which generates a MF trajectory consistent with the optimal policy $\Psi(\bar{X})$ as obtained above.
\end{enumerate}

\par Then, we have the following definition.
\begin{definition}(Mean-field equilibrium \cite{uz2020reinforcement})
The pair $(\mu^*,\bar{X}^*) \in \mathcal{M} \times \mathcal{X} $ is an MFE if $\mu^* = \Psi(\bar{X}^*)$ and $\bar{X}^* = \Lambda(\mu^*)$. More precisely, $\bar{X}^*$ is a fixed point of the map $\Lambda \circ \Psi$.
\end{definition}
The trajectory $\bar{X}^*$ is the MF trajectory at equilibrium with $\mu^*$ as the equilibrium control policy.
The aim now is to design an optimal tracking control policy for \eqref{systemMF} minimizing \eqref{LQTMF} under the information structure discussed above.
\par Before proceeding, we point out that with some abuse of notation, we retain the same notations (as in Figure 1) of a generic agent's control system except by removing the superscript $i$. We now construct the fully observed problem (in Proposition 1 below) using decoder estimates ($Y_k$) from the noisy observations output from the channel. This will follow from two results, namely: the control policy is free of dual-effect \cite{bar1974dual,feldbaum1961dual}, and the optimal control policy is certainty equivalent, as we prove next. Define the error between the plant and the decoder output as:
\begin{align}\label{error_def}
\bar{e}_k:= X_k-Y_k= \left\{
\begin{array}{ll}
 e_k, & \gamma_k=1 \\
 \delta_k, & \gamma_k=0 
 \end{array},
\right.
\end{align} where from \eqref{Recursive}, we let $e_k:= X_k-\mathbb{E}\{X_k|I_k^{d},d_{k}\}$ and $\delta_k:= X_k-AY_{k-1} - BU_{k-1}$, as the error at the transmission and non-transmission times, respectively. 
Then, we have the following lemma.
\begin{lemma}\label{L1}
Consider the system \eqref{systemMF} for the generic agent under the information structure as in the previous section. Then:
\begin{enumerate}[(i)]
\item The relative error $\bar{e}_k$ is independent of all control choices for all $k$.
\item The expected value of $\bar{e}_k$ and the conditional correlation between $Y_k$ and $\bar{e}_k$ given $\mathcal{F}^d_k$, are both zero.
\end{enumerate}
Consequently, the control is free of dual-effect.
\end{lemma}
\begin{proof}
We prove the Lemma in two parts, first for transmission and then for non-transmission times.

\noindent \textbf{Part 1:}  Consider a transmission instant $l_r \in \mathcal{K}(k)$. Then, we have
\begin{align*}
\bar{e}_{l_{r+1}} = e_{l_{r+1}} &= X_{l_{r+1}}-Y_{l_{r+1}}, \\
&= X_{l_{r+1}}-\mathbb{E}\left\lbrace X_{l_{r+1}} |I^d_{l_{r+1}}\right\rbrace,\\
&= A(\phi)e_{l_{r+1}-1} + W_{l_{r+1}-1} \\ & \qquad - \mathbb{E}\left\lbrace A(\phi)e_{l_{r+1}-1} + W_{l_{r+1}-1}|I^d_{l_{r+1}}\right\rbrace.
\end{align*}
Using simple manipulations, the error at the $l_{r+1}-$th transmission time can be recursively in terms of the preceding one as
\begin{align}\label{e_tl+1}
e_{l_{r+1}} &= \eta_1(\phi) e_{l_{r}} + \eta_2(\phi) W_{l_{r}:l_{r+1}-1} \nonumber \\
& - \mathbb{E}\left\lbrace\eta_1(\phi)e_{l_{r}} + \eta_2(\phi) W_{l_{r}:l_{r+1}-1}|I^d_{l_{r+1}}\right\rbrace,
\end{align}
for matrices $\eta_1(\phi)$ and $\eta_2(\phi)$ of appropriate dimensions.
\par Now, we prove that $e_{l_{r}}$ is independent of control actions by induction on the parameter $l_{r}$. Note that the first event time is assumed to be $l_0=0$. Fix $Y_{l_0} = 0$. Then, $e_{l_0} = X_{l_0}$ is independent of all control actions. Now, assume that $e_{l_r}$ is independent of all controls. Now, by assumption, we have that the process noise is independent of controls, $\forall l_r$. Consequently, from \eqref{e_tl+1} and the induction hypothesis, we have that $e_{l_{r+1}}$ is independent of all controls. Finally, by the principle of mathematical induction, (i) holds for Part 1.

\par Next, since $e_{l_r} = X_{l_{r}}-\mathbb{E}\{ X_{l_{r}} |I^d_{l_{r}}\}$, we have $\mathbb{E}\{e_{l_{r}} |I^d_{l_{r}}\} = 0.$ Also, $\mathbb{E}\{ Y_{l_r}e_{l_{r}}' |I^d_{l_{r}}\} =\mathbb{E} \{ Y_{l_r}e_{l_{r}}' |I^d_{l_{r}},Y_{l_r}\} = Y_{l_r}\mathbb{E}\{ e_{l_{r}}' |I^d_{l_{r}}\} = 0,$ where the second equality follows since $Y_{l_r} = \mathscr{D}_{l_r}(I^d_{l_{r}})$ is $\sigma\left(I^d_{l_{r}} \right)$-measurable (the sigma-algebra generated by $I^d_{l_{r}}$). Hence, (ii) holds.

\noindent \textbf{Part 2:} Now, we prove (i) and (ii) for the non-transmission times. Suppose $l_r$ and $l_{r+1}$ are some consecutive transmission times. Then, for any $k \in (l_r, l_{r+1})$, $\bar{e}_k = \delta_k = \bar{\eta}_1(\phi)e_{l_r} + \bar{\eta}_2(\phi)W_{l_r:k-1}$, for appropriate matrices $\bar{\eta}_1(\phi)$ and $\bar{\eta}_2(\phi)$. Then, by Part 1, (i) holds for all non-transmission times.
%
\par Next, we prove (ii). It is easy to see that the information state of the decoder is updated with \textit{new} information only at the transmission instants. More specifically, any estimate between two transmission times $l_r$ and $l_{r+1}$, for some $r$, can be recovered from its information state $I^d_{l_r}$. Thus, 
\begin{align*}
\mathbb{E}\left\lbrace\delta_k|I^d_k\right\rbrace &= \mathbb{E}\left\lbrace\bar{\eta}_1(\phi)e_{l_r} + \bar{\eta}_2(\phi)W_{l_r:k-1}|I^d_k\right\rbrace \\
& = \mathbb{E}\left\lbrace\bar{\eta}_1(\phi)e_{l_r} + \bar{\eta}_2(\phi)W_{l_r:k-1}|I^d_{l_r}\right\rbrace \\
& = \mathbb{E}\left\lbrace\bar{\eta}_1(\phi)e_{l_r}|I^d_{l_r}\right\rbrace + \mathbb{E}\left\lbrace\bar{\eta}_2(\phi)W_{l_r:k-1}|I^d_{l_r}\right\rbrace \\
&= 0.
\end{align*}
The last equality follows from Part 1 and the fact that $W_k$ is an independent zero mean process.
Finally, $\mathbb{E}[Y_k\delta_k' |I^d_k] = 0$ follows in the same manner as in Part 1. Thus (ii) holds and the proof is complete.
\end{proof}
Note that the proof of Lemma 1 is made possible due to the information structure of the scheduler, encoder and decoder. Since the information maps of the scheduler and the controller entail partially nested $\sigma$-algebras, the scheduler is able to recover the controller output information at its own end. This, in addition to the deterministic nature of control policies, allows the scheduler to compute the scheduling policy based on innovations and consequently take complete authority over transmission of \textit{new} information. As a result, the control is dual-effect free.


\par Next, we prove that due to the no dual-effect, the separation principle holds for the underlying tracking problem. Consequently, the optimal control law under the information structure in Table \ref{tab:Infostates}, is certainty equivalent \cite{ramesh2013design}.

%
\begin{theorem}\label{Th1}
Consider the information structure on the generic agent as in Table \ref{tab:Infostates}. Then, the control design problem separates into designing a state decoder and a certainty equivalent controller.
\end{theorem}
\begin{proof}
Consider the cost-to-go as follows:
\begin{align*}
&\mathbb{E}\{(X_k-\bar{X}_k)'Q(\phi)(X_k-\bar{X}_k)+U_k'R(\phi)U_k|I_k^{\mu}\} \\
=&\mathbb{E}\left\lbrace(Y_k-\bar{X}_k+\bar{e}_k)'Q(\phi)(Y_k-\bar{X}_k+\bar{e}_k) \right. \\& \left. + U_k'R(\phi)U_k|I_k^{\mu}\right\rbrace,\\
=&(Y_k-\bar{X}_k)'Q(\phi)(Y_k-\bar{X}_k) + \mathbb{E}\{(Y_k-\bar{X}_k)'Q(\phi)e_k|I_k^{\mu}\} \\& +\mathbb{E}\{\bar{e}_k'Q(\phi)(Y_k-\bar{X}_k)|I_k^{\mu}\} + \mathbb{E}\{\bar{e}_k'Q(\phi)\bar{e}_k|I_k^{\mu}\} \\& + \mathbb{E}\{U_k'R(\phi)U_k|I_k^{\mu}\} , \\
=&(Y_k-\bar{X}_k)'Q(\phi)(Y_k-\bar{X}_k) + \mathbb{E}\{\bar{e}_k'Q(\phi)\bar{e}_k|I_k^{\mu}\} \\& + \mathbb{E}\{U_k'R(\phi)U_k|I_k^{\mu}\} , \\
=&(Y_k-\bar{X}_k)'Q(\phi)(Y_k-\bar{X}_k) + \mathbb{E}\{U_k'R(\phi)U_k|I_k^{\mu}\} \\& + Tr\{Q(\phi)\Delta_k\},
\end{align*}
where $\Delta_k = \left\{
\begin{array}{ll}
 \mathbb{E}\{e_ke_k'|Y_{0:k}\}, & \gamma_k=1 \\
 \mathbb{E}\{\delta_k\delta_k'|Y_{0:k}\}, & \gamma_k=0 
 \end{array}.
\right.$
The third equality follows from Lemma \ref{L1} and the fact that $\bar{X}_k$ is deterministic and independent of $\bar{e}_k$. The last inequality follows again since the relative error is independent of control actions by Lemma \ref{L1}. Then, since the error-induced term $Tr\{Q(\phi)\Delta_k\}$ is independent of controls, the proof is complete.
\end{proof}

Certainty equivalence of the optimal control policy is a consequence of no dual-effect \cite{bar1974dual,feldbaum1961dual} of control as in Lemma \ref{L1}. When the control is free of dual-effect, the covariance of the estimation error is independent of the control signals used.  Thus, the controller can no longer benefit from probing the scheduler for information and can be designed independently of the scheduler and the decoder.
We are now ready to state the following Proposition.
\begin{proposition}[Separated Stochastic Optimal Control Problem]\label{Prop1}
Using Theorem \ref{Th1}, the fully observed system can be constructed using the partially observed one (due to the presence of noisy channel) as: 
\begin{align}\label{Sep_sys}
Y_{k+1} =\left\{
\begin{array}{ll}
 A(\phi)Y_{k} +B(\phi)U_k + \bar{W}_k, & \gamma_{k+1}=1 \\
 A(\phi)Y_{k} +B(\phi)U_k, & \gamma_{k+1}=0
 \end{array},
\right.
\end{align}where $\bar{W}_k = A(\phi)\bar{e}_k + W_k - \bar{e}_{k+1}$, with the associated cost-to-go $(Y_k-\bar{X}_k)'Q(\phi)(Y_k-\bar{X}_k) + \mathbb{E}\{U_k'R(\phi)U_k|I_k^{\mu}\} + Tr\{Q(\phi)\Delta_k\}.$ Then, under Assumption \ref{Ass3},
\begin{enumerate}[(i)]
\item The optimal control policy for the separated problem \eqref{Sep_sys} is given as
\begin{align}\label{Cntrl}
U^{*}_k &= -\Pi(\phi)Y_k - \Gamma(\phi)g_{k+1}
\end{align}
where $\Gamma(\phi) = (R(\phi)+B(\phi)'K(\phi)B(\phi))^{-1}B(\phi)'$, $\Pi(\phi) = \Gamma(\phi)K(\phi)A(\phi)$. Further, $K(\phi)$ is the unique positive definite solution to the algebraic Riccati equation
\begin{align}\label{RE1}
K(\phi) = A(\phi)'K(\phi)A(\phi) - & A(\phi)'K(\phi)'B(\phi)\Pi (\phi)\nonumber \\ &+ Q(\phi)
\end{align} and the trajectory $g_k$ is given as
\begin{align}\label{RE2}
g_k = H(\phi)'g_{k+1}- Q(\phi)\bar{X}_k,
\end{align} where $H(\phi)':= A(\phi)'[I-K(\phi)B(\phi)\Gamma (\phi)]$ is Hurwitz.
\end{enumerate}
\end{proposition}
\begin{proof}
The separated problem follows from \eqref{systemMF} and \eqref{error_def} and Theorem \ref{Th1}. The rest of the proof follows from \cite{bertsekas2000dynamic} and \cite{moon2014discrete} and is thus omitted.
\end{proof}
\subsection{MFE Analysis}
Using Proposition 1, we now prove the existence and uniqueness of the MFE by introducing an operator $\mathcal{T}$ as shown in this section. First note from \eqref{RE2} that $H(\phi) = A(\phi)-B(\phi)\Pi(\phi)$. Then, substituting \eqref{Cntrl} in \eqref{Sep_sys}, we arrive at the closed-loop system as:
\begin{align*}
& Y_{k+1} = \left\{
\begin{array}{ll}
  H(\phi)Y_{k} - B(\phi)\Gamma(\phi)g_{k+1} + \bar{W}_k, &\gamma_{k+1} = 1\\
 H(\phi)Y_{k} - B(\phi)\Gamma(\phi)g_{k+1}, &\gamma_{k+1} = 0
 \end{array}.
\right.
\end{align*}
Then, using $\bar{W}_k$ from Proposition \ref{Prop1} and \eqref{error_def}, we can rewrite the above closed-loop system as 
\begin{align}\label{X1}
X_{k+1} &= H(\phi)X_k - B(\phi)\Gamma(\phi)g_{k+1} + B(\phi)\Pi(\phi)\bar{e}_k + W_k.
\end{align}
It is to be noted that $g_k \in \mathcal{X}$, if $g_0 = -\sum_{j=0}^{\infty}{(H^j(\phi))'}Q(\phi)\bar{X}_j$, which further gives $g_k = -\sum_{j=k}^{\infty}{(H^{j-k}(\phi))'}Q(\phi)\bar{X}_j$. Substituting this value of $g_k$ in \eqref{X1}, we get
\begin{align*}
X_{k+1} &= H(\phi)X_k + B(\phi)\Gamma(\phi)\sum_{j=k+1}^{\infty}{(H^{j-k-1}(\phi))'}Q(\phi)\bar{X}_j \\&+ B(\phi)\Pi(\phi)\bar{e}_k + W_k.
\end{align*}
Now, taking expectation on both sides and denoting $\hat{X}_k(\phi) = \mathbb{E}\{X_k\}$ as the aggregate trajectory of agents of type $\phi$, we get 
\begin{align}\label{X2}
& \hat{X}_{k+1}(\phi)= \\ & H(\phi)\hat{X}_k(\phi) + B(\phi)\Gamma(\phi)\sum_{j=k+1}^{\infty}{(H^{j-k-1}(\phi))'Q(\phi)\bar{X}_j}, \nonumber
\end{align}
where we use the tower property of conditional expectation and Lemma \ref{L1}, to get $\mathbb{E}\{\bar{e}_k\} = \mathbb{E}\{\mathbb{E}\{\bar{e}_k|I_{k}^d\}\}= 0$. Finally, \eqref{X2} can be simplified further as:
\begin{align}
& \hat{X}_{k}(\phi) = H^k(\phi)\nu_{\phi,0}\\ &+ \sum_{j=0}^{k-1}{H^{k-j-1}(\phi)B(\phi)\Gamma(\phi)\sum_{s=j+1}^{\infty} (H^{s-j-1}(\phi))'Q(\phi)\bar{X}_s}. \nonumber
\end{align}
\par Now, using the emprical distribution \eqref{Emp_D}, define the operator $\mathcal{T}$ as:
\begin{align}\label{MFS}
\mathcal{T}(\bar{X})(k) := \sum_{\phi \in \Phi}{\hat{X}_{k}(\phi)P(\phi)} ,
\end{align} where $\mathcal{T}(\bar{X})(\cdot)$ maps the input sequence to another sequence at time $k$. Using this operator, we prove existence and uniqueness of the  equilibrium MF trajectory by finding the fixed point of \eqref{MFS}, under the following assumption.
\begin{assumption}\label{Ass4}
We assume $\Xi := \|H(\phi)\| +\zeta <1, \forall \phi$, where $\zeta := \sum_{\phi \in \Phi}{\frac{\|Q(\phi)\|\|B(\phi)\Gamma(\phi)\|}{(1-\|H(\phi)\|)^2}P(\phi)} $.
\end{assumption}
Assumption \ref{Ass4} is motivated by results from existing literature \cite{huang2007large,moon2014discrete,uz2020reinforcement}. While it is stronger than the ones in \cite{uz2020approximate,moon2014discrete}, it entails linear MF trajectory dynamics, which is easily tractable.
\begin{theorem}\label{Th2}
Under Assumptions \ref{Ass3}-\ref{Ass4}, the following hold true:
\begin{enumerate}[(i)]
\item The operator $\mathcal{T}(\bar{X}) \in \mathcal{X}, ~~\forall \bar{X} \in \mathcal{X}$. Furthermore, there exists unique $\bar{X}^* \in  \mathcal{X} $ such that $\mathcal{T}(\bar{X}^*) = \bar{X}^*$.
\item $\bar{X}_{k}^*$ follows linear dynamics, i.e., $\exists ~ L^* \in \mathcal{L}:= \{L\in \mathbb{R}^{n\times n}:~ \|L\|\leq 1, \bar{X}_{k+1}^* = L\bar{X}_k^*\}$, where $\bar{X}^*_k$ is the aggregate trajectory of the agents at equilibrium, and $\bar{X}_0^* = \sum_{\phi \in \Phi}{\nu_{\phi,0}P(\phi)}$.
\end{enumerate}
\end{theorem}
\begin{proof}
\begin{enumerate}[(i)]
\item Consider the linear system \eqref{X2} with driving input $\bar{X}_k$, which is bounded since $\bar{X} \in \mathcal{X}.$ Since $\|H(\phi)\| <1$ by Assumption \ref{Ass4}, and $g_k \in \mathcal{X},~\forall k$, we have from \eqref{X2} and \eqref{MFS}, that $\sup_{k\geq 0}{\|\mathcal{T}(\bar{X})(k)\|} < \infty$, which proves the first statement in part (i). Next, consider the following:
\begin{align*}
&\|\mathcal{T}(\bar{X}_1)- \mathcal{T}(\bar{X}_2)\|_{\infty} = \|\sum_{\phi \in \Phi}({\hat{X}_{1} - \hat{X}_{2})P(\phi)}\|_{\infty} \\
& \leq \sum_{\phi \in \Phi}{(\|Q(\phi)\|\|B(\phi)\Gamma(\phi)\|\left(\sum_{s=0}^{\infty}\|H(\phi)\|^s\right)^2P(\phi)} \\ & \qquad \times \|\bar{X}_1 - \bar{X}_2\|_{\infty} \\
& = \zeta \|\bar{X}_1 - \bar{X}_2\|_{\infty},
\end{align*} 
where the last equality follows from Assumption \ref{Ass4}. Finally, using Banach's fixed point theorem and the first statement of part (i), $\mathcal{T}(\bar{X})$ has a unique fixed point in $\mathcal{X}$.
\item Define the operator $\bar{\mathcal{T}}: \mathbb{R}^{n\times n} \rightarrow \mathbb{R}^{n\times n}$, given as:
\begin{align*}
\bar{\mathcal{T}}_{\phi}(L) & := H(\phi) + B(\phi)\Gamma(\phi)\sum_{\alpha=0}^{\infty}{(H^{\alpha}(\phi))'Q(\phi)L^{\alpha+1}},\\
\bar{\mathcal{T}}(L) & := \sum_{\phi \in \Phi} \bar{\mathcal{T}}_{\phi}(L) P(\phi)
\end{align*} with $L^* = \bar{\mathcal{T}}(L^*)$ and $\hat{X}_{k+1}^* = L^*\hat{X}_k^*$. To prove that such a $L^*$ indeed exists, we follow the same lines of proof as in \cite{uz2020reinforcement} to arrive at
\begin{align*}
	& \|\bar{\mathcal{T}}(L_2) - \bar{\mathcal{T}}(L_1)\| \\
	& \hspace{0.6cm} < \sum_{\phi \in \Phi}\frac{\|B(\phi)\Gamma(\phi)\|\|Q\|}{(1-\|H(\phi)\|)^2}\|L_2-L_1\| P(\phi),
\end{align*}
which, under Assumption \ref{Ass4}, establishes that $\bar{\mathcal{T}}$ is a contraction. Using completeness of $\mathcal{L}$ and Banach's fixed point theorem, we indeed have the existence of such an $L^*$. Finally, from (i) above, we get that the unique MF trajectory $\bar{X}^*$ can be constructed recursively as $\bar{X}^*_{k} = (L^*)^k\bar{X}^*_0$ with $\hat{X}_0 = \nu_{\phi,0}$.

\end{enumerate}
\end{proof}
Theorem \ref{Th2} (i) provides us with a unique MFE while the linearity of the MF trajectory in (ii) gives a control law which is linear in the state of the agent and the equilibrium trajectory \cite{uz2020reinforcement}. This further makes the computation of this trajectory tractable which would otherwise have involved a non-causal infinite sum.
\subsection{$\epsilon$-Nash Equilibrium}
We now show that under the information structure in Table \ref{tab:Infostates}, the MFE constitutes an $\epsilon$-Nash equilibrium for the $N$-player game. 
Before doing that, we first provide the definition of $\epsilon$-Nash equilibrium for the discrete-time MFG under communication constraints. Let us denote the space of admissible centralized control policies for agent $i$ as $\bar{\mathcal{M}}_i^c$, under a centralized information structure, where each agent is assumed to have  access to other agents' output histories. Then, we have the following:
\begin{definition}\cite{moon2014discrete,saldi2018markov}
The set of control policies $\{\mu^i \in {\mathcal{M}}_i^c, ~i \in [1,N]\}$ constitute an $\epsilon$-Nash equilibrium with respect to the cost functions $\{J_i^N, ~i \in [1,N]\}$, if, for some $\epsilon > 0$,
\begin{align*}
J_i^N(\mu^i,\mu^{-i}) -\epsilon \leq \inf_{\pi^i \in {\mathcal{M}_i^c}}{J_i^N(\pi^i,\mu^{-i})}.
\end{align*}
\end{definition}
We start by proving that the mass behaviour in \eqref{LQT} converges to the equilibrium MF trajectory as in Theorem \ref{Th2} (i). Note henceforth, signals superscripted by an asterisk ($*$) will represent quantities in the equilibrium, e.g., $X^{i*}_k$ denotes the state of agent $i$ at time $k$ under equilibrium control policy.
\begin{lemma}\label{L2}
Suppose Assumptions \ref{Ass3}-\ref{Ass4} hold and all the agents operate under the equilibrium control policy. Then, the coupling term in \eqref{LQT} converges (in the mean-square sense) to the equilibrium mean-field trajectory, i.e.,
\begin{align}\label{MSE_Con}
\lim\limits_{N\rightarrow \infty}\limsup\limits_{T\rightarrow \infty}\frac{1}{T}\mathbb{E}\left\lbrace \sum_{k=0}^{T-1}\Bigg\|\frac{1}{N}\sum_{i=1}^{N}{X_k^{i*}} - \bar{X}_k^*\Bigg\|^2\right\rbrace = 0.
\end{align}
\end{lemma}
\begin{proof}
First, consider the following: 
\begin{align*}
& \lim\limits_{N\rightarrow \infty}\limsup\limits_{T\rightarrow \infty}\frac{1}{T}\mathbb{E}\left\lbrace \sum_{k=0}^{T-1}\Bigg\|\frac{1}{N}\sum_{i=1}^{N}{X_k^{i*}} - \bar{X}_k^*\Bigg\|^2\right\rbrace \\
& \leq \lim\limits_{N\rightarrow \infty}\limsup\limits_{T\rightarrow \infty}\frac{2}{T}\mathbb{E}\left\lbrace \sum_{k=0}^{T-1}\Bigg\|\frac{1}{N}\sum_{i=1}^{N}{X_k^{i*}} -{\hat{X}_k^{*}(\phi_i)}\Bigg\|^2\right\rbrace \\
& ~~~~+ 2\lim\limits_{N\rightarrow \infty} \sup_{k\geq 0}\Bigg\|\frac{1}{N}\sum_{i=1}^{N}{\hat{X}_k^{*}(\phi_i)} - \bar{X}^*_k\Bigg\|^2 
\end{align*}
where the inequality follows from the identity $\|a+b\|^2 \leq 2\|a\|^2+2\|b\|^2$ and $\frac{1}{N}\sum_{i=1}^{N}{\hat{X}_k^{*}(\phi_i)} = \sum_{\phi_i \in \Phi}{\hat{X}^*_k(\phi_i)P_N(\phi_i)}, \forall k$.
We now prove that the first term in the above inequality vanishes. Let $Z_k^{i*} = X_k^{i*} - \hat{X}_k^{i*}$. Then, using \eqref{X1}, we have
\begin{align*}
Z_{k+1}^{i*} = H(\phi_i)Z_{k}^{i*} + W_k^i + B(\phi_i)\Pi(\phi_i)\bar{e}_k^i.
\end{align*}
Then, as in \cite[Lemma 2]{moon2014discrete}, $\exists ~T_1 > 0$ and $M_1 > 0,$ independent of $T$ and $N$, such that $\frac{2}{T}\mathbb{E}\left\lbrace \sum_{k=0}^{T-1}\|\frac{1}{N}\sum_{i=1}^{N}{Z_k^{i*}}\|^2\right\rbrace \leq \frac{M_1}{N},~ \forall T > T_1$, which implies that $\limsup\limits_{T\rightarrow \infty}\frac{2}{T}\mathbb{E}\left\lbrace \sum_{k=0}^{T-1}\|\frac{1}{N}\sum_{i=1}^{N}{Z_k^{i*}}\|^2\right\rbrace \leq \frac{M_1}{N}$. This finally gives $\lim\limits_{N\rightarrow \infty}\limsup\limits_{T\rightarrow \infty}\frac{2}{T}\mathbb{E}\left\lbrace \sum_{k=0}^{T-1}\|\frac{1}{N}\sum_{i=1}^{N}{Z_k^{i*}}\|^2\right\rbrace = 0.$ 
Next, since the support of $\Phi$ is finite (and hence compact), this implies that $\hat{X}^*_k(\phi_i)$ is uniformly bounded for all $k$. Further, since $\bar{X}^* \in \mathcal{X}$ from Theorem 2, we have $\lim\limits_{N\rightarrow \infty} \sup_{k\geq 0}\|\frac{1}{N}\sum_{i=1}^{N}{\hat{X}_k^{*}(\phi_i)} - \bar{X}^*_k\|^2  = 0$. Thus, \eqref{MSE_Con} holds and the proof is complete.
\end{proof}
Now, we prove the $\epsilon$-Nash property of the MFE. Toward that end, consider the following cost functions: 
\begin{align}
&J_i^N(\mu^{i*}, \mu^{-i*})= \\& \limsup_{T\rightarrow \infty}\frac{1}{T}\mathbb{E}\left\lbrace \sum_{k=0}^{T-1}\Bigg\|X_k^{i*}-\frac{1}{N}\sum_{j=1}^{N}{X_k^{j*}}\Bigg\|^2_{Q} + \|U_k^{i*}\|^2_{R}\right\rbrace, \nonumber \\
&J(\mu^{i*}, \bar{X}^*)= \\& \limsup_{T\rightarrow \infty}\frac{1}{T}\mathbb{E}\left\lbrace \sum_{k=0}^{T-1}\|X_k^{i*}-\bar{X}^*\|^2_{Q} + \|U_k^{i*}\|^2_{R}\right\rbrace,\nonumber \\
&J_i^N(\pi^i, \mu^{-i*})= \\& \limsup_{T\rightarrow \infty}\frac{1}{T}\mathbb{E}\left\lbrace \sum_{k=0}^{T-1}\Bigg\|X_k^{i,\pi^i}-\frac{1}{N}\sum_{j=1}^{N}{X_k^{j*}}\Bigg\|^2_{Q} + \|V_k^{i}\|^2_{R}\right\rbrace, \nonumber 
\end{align}
where $\bar{X}^*$ is the equilibrium MF trajectory (Theorem \ref{Th2}) and $X_k^{i,\pi^i}$ is the state of agent $i$ at time $k$ when it chooses a control law $\pi^i$ from the set of centralized policies ${\mathcal{M}}_i^c$. Notice that this set is strictly larger than the set $\mathcal{M}$. Furthermore, the control action $V_k^i$ is derived from $\pi^i \in {\mathcal{M}}_i^c$.
Now, we have the following theorem stating the $\epsilon$-Nash result, i.e., that the control laws prescribed by the MFE are also $\epsilon$-Nash in the finite
population case.

\begin{theorem}\label{Th3}
Under the Assumptions \ref{Ass3}-\ref{Ass4}, the set of $N$ decentralized control laws $\{\mu^{i*},~i \in [1,N]\}$, where $\mu^{i*} = \mu^*$, constitutes an $\epsilon$-Nash equilibrium for the LQ-MFG with communication constrained AWGN channel, more precisely, we have
\begin{align}\label{epsNash}
J_i^N(\mu^{i*},\mu^{-i*}) & \leq \inf_{\pi^i \in {\mathcal{M}}_i^c}{J_i^N(\pi^i,\mu^{-i*})} + \mathcal{O}\left(\limsup_{T\rightarrow \infty}\sqrt{\epsilon_T^N}\right),
\end{align}
where $\epsilon_T^N = \frac{1}{T}\mathbb{E}\left\lbrace \sum_{k=0}^{T-1}{\|\frac{1}{N}\sum\limits_{j=1}^{N}{X_k^{j*}} - \bar{X}_k^*\|^2} \right\rbrace$.
\end{theorem}
\begin{proof}
We prove the theorem in two steps. In the first step, we derive an upper bound on $J_i^N(\mu^{i*}, \mu^{-i*})- J(\mu^i, \bar{X}^*)$, and in step 2, on $J(\mu^i, \bar{X}^*) - J_i^N(\pi^i, \mu^{-i*})$. Finally, we combine the two to get \eqref{epsNash}.

\noindent \textbf{Step 1:} Consider the following:
\begin{align}
&J_i^N(\mu^{i*}, \mu^{-i*})- J(\mu^i, \bar{X}^*) \leq J_i^N(\mu^{i*}, \mu^{-i*})- J(\mu^{i*}, \bar{X}^*) \nonumber \\
& = \limsup_{T\rightarrow \infty}\frac{1}{T}\mathbb{E}\left\lbrace \sum_{k=0}^{T-1}\Bigg\|X_k^{i*}-\frac{1}{N}\sum_{j=1}^{N}{X_k^{j*}}\Bigg\|^2_{Q} \right. \nonumber \\& \left. ~~~~ - \|X_k^{i*}-\bar{X}_k^*\|^2_{Q} \right\rbrace, \nonumber \\
& = \limsup_{T\rightarrow \infty}\frac{1}{T}\mathbb{E}\left\lbrace \sum_{k=0}^{T-1}\Bigg\|\bar{X}_k^*-\frac{1}{N}\sum_{j=1}^{N}{X_k^{j*}}\Bigg\|^2_{Q}\right. \nonumber \\ &\left. ~~~~+ 2 \left( X_k^{i*}-\bar{X}_k^*\right)'Q\left( \bar{X}_k^*-\frac{1}{N}\sum_{j=1}^{N}{X_k^{j*}}\right) \right\rbrace \nonumber \\
& \leq \limsup_{T\rightarrow \infty} \|Q\|\epsilon_T^N + 2\|Q\|\mathbb{E}\left\lbrace \sqrt{\frac{1}{T}\sum_{k=0}^{T-1}\|X_k^{i*}-\bar{X}_k^*\|^2} \right. \nonumber \\
&\left. \qquad \qquad \times \sqrt{\frac{1}{T}\sum_{k=0}^{T-1}\Bigg\|\bar{X}_k^*-\frac{1}{N}\sum_{j=1}^{N}{X_k^{j*}}\Bigg\|^2} \right\rbrace \nonumber \\
& \leq \limsup_{T\rightarrow \infty}\|Q\|\epsilon_T^N \nonumber \\& + \limsup_{T\rightarrow \infty}2\|Q\|\sqrt{\epsilon_T^N\mathbb{E}\left\lbrace \frac{1}{T}\sum_{k=0}^{T-1}\|X_k^{i*}-\bar{X}_k^*\|^2 \right\rbrace} \label{X3}
\end{align}
where both the inequalities follow from the Cauchy-Schwarz inequality.

\noindent \textbf{Step 2: }Consider the following:
\begin{align*}
&J_i^N(\pi^i, \mu^{-i*})= \\& \limsup_{T\rightarrow \infty}\frac{1}{T}\mathbb{E}\left\lbrace \sum_{k=0}^{T-1}\Bigg\|X_k^{i,\pi^i}-\frac{1}{N}\sum_{j=1}^{N}{X_k^{j*}}\Bigg\|^2_{Q} + \|V_k^{i}\|^2_{R}\right\rbrace \\
\end{align*}
\begin{align*}
& = \limsup_{T\rightarrow \infty}\frac{1}{T}\mathbb{E}\left\lbrace \sum_{k=0}^{T-1}\|X_k^{i,\pi^i}-\bar{X}_k^{i*}\|^2_Q + \|V_k^{i}\|^2_{R} \right. \\& \left. \qquad \qquad + \Bigg\|\frac{1}{N}\sum_{j=1}^{N}{X_k^{j*}} - \bar{X}_k^{i*}\Bigg\|^2_Q \right. \\& \left. + 2 \left(X_k^{i,\pi^i}-\bar{X}_k^* \right)'Q \left( \bar{X}_k^*-\frac{1}{N}\sum_{j=1}^{N}{X_k^{j*}}\right) \right\rbrace \\
& = J(\pi^i, \bar{X}^*) - \limsup_{T\rightarrow \infty}\frac{1}{T}\mathbb{E} \left\lbrace \sum_{k=0}^{T-1} \Bigg\|\frac{1}{N}\sum_{j=1}^{N}{X_k^{j*}} - \bar{X}_k^{i*}\Bigg\|^2_Q \right. \\ & \left. + 2 \left(\bar{X}_k^* - X_k^{i,\pi^i} \right)'Q \left( \bar{X}_k^*-\frac{1}{N}\sum_{j=1}^{N}{X_k^{j*}}\right) \right\rbrace \\
& \geq J(\mu^i, \bar{X}^*) - \limsup_{T\rightarrow \infty}\frac{1}{T}\mathbb{E} \left\lbrace \sum_{k=0}^{T-1} \Bigg\|\frac{1}{N}\sum_{j=1}^{N}{X_k^{j*}} - \bar{X}_k^{i*}\Bigg\|^2_Q \right. \\ & \left. + 2 \left(\bar{X}_k^* - X_k^{i,\pi^i} \right)'Q \left( \bar{X}_k^*-\frac{1}{N}\sum_{j=1}^{N}{X_k^{j*,\pi^i}}\right) \right\rbrace.
\end{align*} 
Finally, using Cauchy-Schwarz inequality in the same manner as for \eqref{X3}, we get
\begin{align}\label{X4}
&J(\mu^i, \bar{X}^*) - J_i^N(\pi^i, \mu^{-i*}) \leq \limsup_{T\rightarrow \infty}\|Q\|\epsilon_T^N \nonumber \\& + \limsup_{T\rightarrow \infty}2\|Q\|\sqrt{\epsilon_T^N\mathbb{E}\left\lbrace \frac{1}{T}\sum_{k=0}^{T-1}\|X_k^{i,\pi^i}-\bar{X}_k^*\|^2 \right\rbrace}. 
\end{align}

\par Similar to Theorem 3 in \cite{moon2014discrete}, there exist $M_2,~T_2>0$, such that $\frac{1}{T}\mathbb{E}\left\lbrace \sum_{k=0}^{T-1}\|X_k^{i*}\|^2 \right\rbrace < M_2,~ \forall T > T_2$. Further, from Theorem \ref{Th2} (i), there exist $M_3, ~T_3>0$, such that $\frac{1}{T}\mathbb{E}\left\lbrace \sum_{k=0}^{T-1}\|\bar{X}_k^{i*}\|^2 \right\rbrace < M_3, ~\forall T>T_3$. Finally, since $\bar{\mathcal{M}}_i \subseteq {\mathcal{M}}_i^c$,  $\inf_{\pi^i \in {\mathcal{M}}_i^c}J_i^N(\pi^i, \mu^{-i*}) \leq J_i^N(\mu^{i*}, \mu^{-i*})$, we may consider $\pi^i \in {\mathcal{M}}_i^c$ such that there exist $M_4,~T_4 >0$, with the property that $\frac{1}{T}\mathbb{E}\left\lbrace \sum_{k=0}^{T-1}\|X_k^{i, \pi^i}\|^2 \right\rbrace < M_4,~\forall T>T_4$. Choose $T_5 = \max\{T_1,T_2,T_3,T_4\}$, and let $T > T_5$, following which, we have \eqref{epsNash} from \eqref{X3} and \eqref{X4}. Finally, define $\epsilon:= \mathcal{O}\big(\limsup_{T\rightarrow \infty}\sqrt{\epsilon_T^N}\big)$, which converges to 0 as $N \rightarrow \infty$ using Lemma \ref{L2}. The proof is thus complete.
\end{proof}
Before concluding this section, we remark that according to Theorem \ref{Th3}, the decentralized equilibrium  policy provides an $\epsilon$-Nash equilibrium for the centralized policy structure in the $N$-player game. Consequently, it also provides an $\epsilon$-Nash equilibrium for the decentralized policy structure in the original $N$-player game formulated in Section II.
\section{Simulations}\label{Sec4}
In this section we demonstrate the performance of MFE under different scheduling policies. We simulate a finite population game with scalar dynamics and a single type $\phi$. The dynamics and cost parameters of the agents satisfy Assumptions \ref{Ass3} and \ref{Ass4}. For Figure, \ref{fig:est_err} we simulate a game of $N = 100$ agents and show that in spite of significant channel noise the estimation error decreases allowing the agents to form a consensus in a very short time. Note that the output does not perfectly mimic the true state due to the asynchronous nature of communication.
\begin{figure}[h]
	\centering
	\begin{subfigure}{0.23\textwidth} 
		\includegraphics[scale=.20]{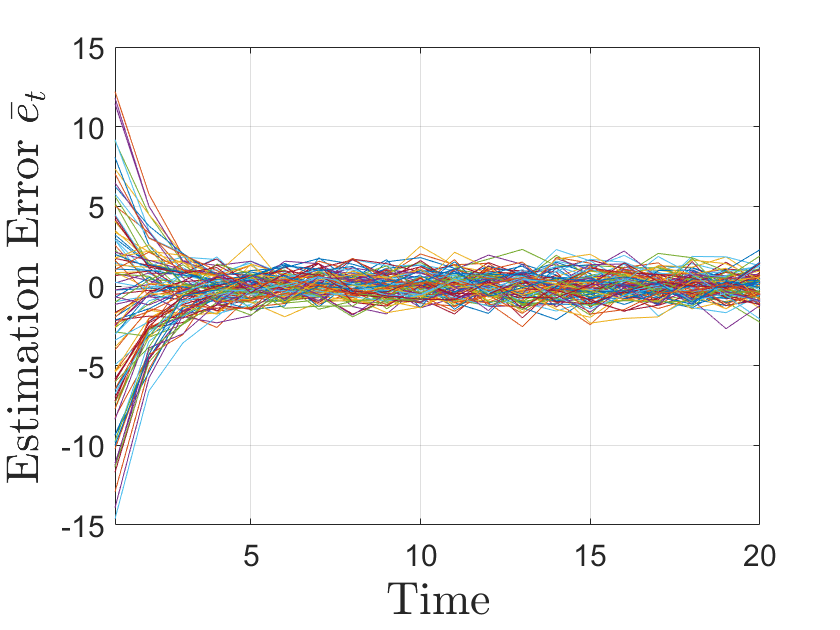}
		\caption{Estimation error}
		\label{fig:est_err}
	\end{subfigure}
	\begin{subfigure}{0.23\textwidth} 
		\includegraphics[scale=.19]{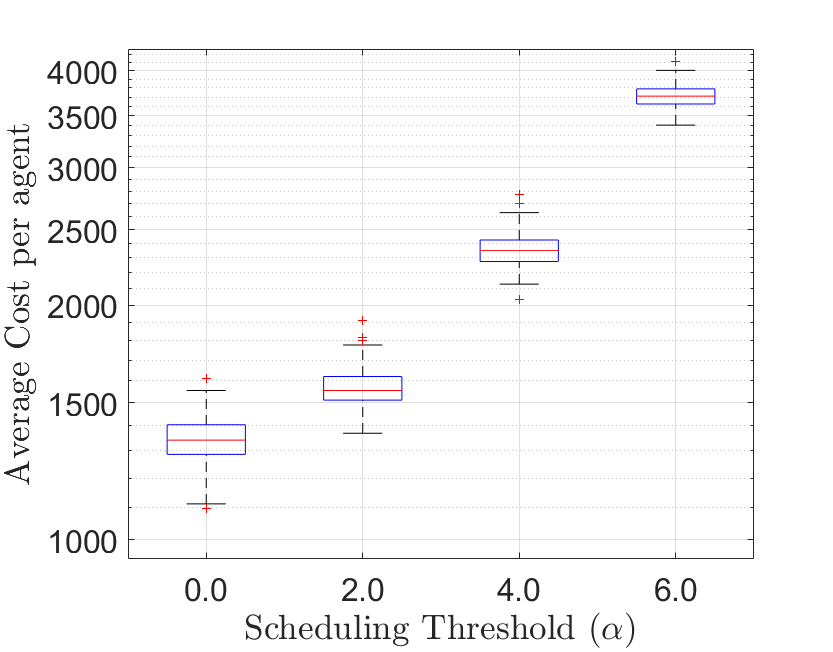}
		\caption{Average cost vs Scheduling Threshold $(\alpha)$}
		\label{fig:sched_pol}
	\end{subfigure}
	\caption{Performance of the equilibrium  policy  under communication constraints}
\end{figure}

For Figure \ref{fig:sched_pol},we simulate the behavior of N=1000 agents and plot the average cost per agent on a logarithmic axis against scheduling threshold of $\alpha = {0.0,2.0,4.0,6.0}$. The figure shows a box plot depicting the median (red line) and spread (box) of the average cost per agent over $100$ runs for each value of $\alpha$. The plot shows a clear increase in average cost per agent as $\alpha$ is increased, indicating that an increase in $\alpha$ leads to an increase in estimation error, which in turn causes a higher average cost per agent. This indicates that a compromise can be reached between performance (average cost) vs communication frequency through a judicious choice of the threshold parameter $\alpha$.

\section{Conclusion}\label{Sec5}
In this paper, we have studied LQ-MFGs under communication constraints, namely when there is intermittent communication over an AWGN channel. Under the defined information structure involving the scheduler, encoder, decoder and controller, we have proved that the control is free of the dual-effect in the mean field limit. Consequently, the optimal control policy has been shown to be certainty equivalent. Under appropriate assumptions, we have established the existence, uniqueness and characterization (linearity) of the mean-field trajectory, shown to have the $\epsilon$-Nash property. We have also empirically demonstrated that the performance of the equilibrium policies deteriorates for decreasing communication frequency, aligned with intuition.
\bibliographystyle{IEEEtran}
\bibliography{references}

\end{document}